\documentclass[11pt, a4paper]{article}
\pdfoutput=1 

\usepackage{fullpage}

\usepackage{mdframed}

\usepackage[utf8]{inputenc}

\usepackage{libertine}

\usepackage{amsfonts}
\usepackage{amsmath}
\usepackage{amssymb}
\usepackage{amsthm}
\usepackage{float}
\usepackage{bm}
\usepackage{enumitem}
\usepackage[svgnames]{xcolor}

\usepackage{tikz}  
\usetikzlibrary{arrows}
\usetikzlibrary{patterns,snakes}
\usetikzlibrary{decorations.shapes}
\tikzstyle{overbrace text style}=[font=\tiny, above, pos=.5, yshift=5pt]
\tikzstyle{overbrace style}=[decorate,decoration={brace,raise=5pt,amplitude=3pt}]
\usetikzlibrary{shapes.geometric}

\usepackage{subcaption}

\usepackage{cmap}
\usepackage[breaklinks=true]{hyperref}
\hypersetup{colorlinks={true},urlcolor={blue},linkcolor={DarkBlue},citecolor=[named]{DarkGreen}}

\usepackage{natbib}

\newcommand{\shortcite}[1]{(\citeyear{#1})}
\setcitestyle{authoryear, open={(},close={)}}

\usepackage{xspace}
\usepackage{fancyhdr}
\usepackage{tcolorbox}

\usepackage{microtype}
\usepackage[capitalise,nameinlink]{cleveref}

\usepackage{doi}

\usepackage{booktabs} 
\usepackage[ruled,linesnumbered]{algorithm2e} 
\usepackage{setspace}

\SetAlFnt{\small}
\SetAlCapFnt{\small}
\SetAlCapNameFnt{\small}
\SetAlCapHSkip{0pt}
\IncMargin{-\parindent}

\usepackage{authblk}
\usepackage{comment}

\usepackage{xspace}
\usepackage{amsmath}
\usepackage{amsthm}
\usepackage{amsfonts}
\usepackage{amssymb}
\usepackage{cleveref}
\usepackage{tikz}
\usepackage{todonotes}
\usepackage{algorithmic}

\newcommand{\EFXr}{EFXr\xspace}

\newcommand{\NP}{\textsf{NP}\xspace}
\newcommand{\NPh}{\NP-hard\xspace}

\newcommand{\NPc}{\NP-complete\xspace}

\newcommand{\PLS}{\textsf{PLS}\xspace}
\newcommand{\partition}{\textsf{PARTITION}\xspace}

\newcommand{\width}{edge-cut width}
\newcommand{\supwidth}{super edge-cut width}
\newcommand{\bigoh}{\mathcal{O}}    

\newcommand{\ecw}{\widthshort}

\newcommand{\widthshort}{\operatorname{ecw}}
\newcommand {\supwidthshort}{\operatorname{sec}}

\newcommand{\loc}{\operatorname{loc}}

\newcommand{\items}{\ensuremath{A}}

\newcommand{\Oh}[1]{{\ensuremath{\mathcal{O}\left(#1\right)}}}

\newtheorem{theorem}{Theorem}

\newtheorem{observation}{Observation}
\newtheorem{corollary}{Corollary}
\newtheorem{definition}{Definition}

\newtheorem{lemma}[theorem]{Lemma}

\newtheorem{claim}{Claim}

\newif\iflong
\newif\ifshort

\longtrue

\iflong
\else
\shorttrue
\fi

\title{EF1 and EFX Orientations}

\author[1]{Argyrios Deligkas}
\author[1]{Eduard Eiben}
\author[1]{Tiger-Lily Goldsmith}
\author[2]{Viktoriia Korchemna}

\affil[1]{Royal Holloway University of London, United Kingdom}
\affil[2]{TU Wien}

\begin{document}

\maketitle

\begin{abstract}
We study the problem of finding fair allocations -- EF1 and EFX -- of indivisible goods with orientations. 
In an orientation, every agent gets items from their own predetermined set.
For EF1, we show that EF1 orientations always exist when agents have monotone valuations, via a pseudopolynomial-time algorithm. 
This surprisingly positive result is the main contribution of our paper.
We complement this result with a comprehensive set of scenarios where our algorithm, or a slight modification of it, finds an EF1 orientation in polynomial time.
For EFX, we focus on the recently proposed graph instances, where every agent corresponds to a vertex on a graph and their allowed set of items consists of the edges incident to their vertex.
It was shown that finding an EFX orientation is NP-complete in general.
We prove that it remains intractable even when the graph has a vertex cover of size 8, or when we have a multigraph with only 10 vertices. 
We essentially match these strong negative results with a fixed-parameter tractable algorithm that is virtually the best someone could hope for.
\end{abstract}

%

\section{Introduction}
\label{sec:intro}
The allocation of a set of {\em indivisible} goods to a set of agents in a way that is considered to be {\em ``fair''} is a problem that has been studied since ancient times. 
Since {\em envy free} allocations -- no agent prefers the bundle of any other agent over their own -- for indivisible goods are not always guaranteed to exist, in recent decades mathematicians, economists, and computer scientists formally studied the problem and have proposed several different fairness solution concepts~\citep{LiptonMMS04,BouveretL08,Budish11,CaragiannisKMPSW19}. 

Arguably, EF1 and EFX are the two solution concepts that have received the majority of attention in the literature and have created a long stream of work.
An allocation is EF1, if it is envy-free up to one good, i.e., any envy from one agent $i$ to some agent $j$ is eliminated by removing a specific item from the bundle of agent $j$. 
On the other hand, an allocation is EFX if it is envy-free up to {\em any} good, i.e., any envy towards agent $j$ is eliminated by removing {\em any} item from $j$'s bundle.
While EFX is a stronger fairness notion, it is unknown whether it always exists; this is one of the main open problems in fair division. On the other hand, EF1 allocations are always guaranteed to exist and in fact, we can efficiently compute such an allocation via the envy-cycle elimination algorithm~\citep{LiptonMMS04}.

However, both EF1 and EFX allow for allocations that can be considered ``counterintuitive'' in the best case, or {\em wasteful} in the worst. 
Consider for example the case where we have two agents, $X$ and $Y$, and three items, $a, b$, and $c$. 
The valuations of $X$ for $a,b,c$ are $1, 1, 0.2$ respectively, while the valuations of $Y$ are $0,1,0$. 
Observe now that the allocation that gives $X$ item $a$ and $Y$ items $b$ and $c$ is both EFX and EF1.
Still, giving item $c$ to agent $Y$ seems rather unreasonable since item $c$ is ``irrelevant'' to agent $Y$! 
Luckily for us, this issue can be fixed by giving item $c$ to $X$ instead. But is such a ``fix'' always possible? 
In other words, does a fair allocation always exist under the constraint that every agent gets goods from a restricted set, i.e., a subset of goods that they approve? This is the question we answer in this paper. 

Our work is inspired by the recent paper of \cite{christodoulou2023fair} that studies valuations on graphs. In that model, an instance of the problem is represented via a graph whose vertices correspond to agents with additive utilities and the edges correspond to goods. 
There, each agent had positive utility only for the goods that corresponded to edges incident to their vertex, i.e., only those goods were {\em relevant} to them. The value of an agent for every other good, non-incident to their vertex, was zero. 

\cite{christodoulou2023fair} studied the existence and complexity of finding EFX allocations and EFX {\em orientations}. An orientation is an allocation where every agent gets {\em only} edges adjacent to them, i.e., every edge is ``oriented'' towards the incident agent that gets it.
\cite{christodoulou2023fair} showed something really interesting. They have shown that albeit EFX allocations always exist for this model and they can be computed in polynomial time, 
EFX orientations fail to exist and in fact, the corresponding problem is \NPc even for binary, additive and symmetric valuations for the agents.

\subsection{Our contribution}
Our contribution is twofold: (a) we initiate the study of EF1 orientations; (b) we examine  EFX orientations through the lens of parameterized complexity.

Our main result is to prove that an EF1 orientation always exists when the valuations of the agents are monotone!
In fact, we prove our result for a more general model than the one from~\citeauthor{christodoulou2023fair}~\shortcite{christodoulou2023fair}, where instead of graphs, we consider hypergraphs, i.e., the goods now correspond to hyperedges. In other words, each agent has a subset of goods that are {\em relevant} to them.
We prove our result algorithmically (Theorem~\ref{thm: EF1_exist}). 
The base of our algorithm is the well-known envy-cycle elimination algorithm~\citep{LiptonMMS04}, although it requires two careful modifications to indeed produce an orientation.
The first modification is almost straightforward: every item is allocated to an agent that is incident to it. 
The second modification is required after we swap the bundles of some agents when we resolve an envy cycle. After the swap, an agent might get goods that are not relevant to them. 
If this is the case, we {\em remove} any irrelevant items from all the bundles of the partial allocation and we redistribute them. However, a priori it is not clear whether this procedure will ever terminate. As we prove via a potential argument, the procedure indeed terminates, albeit in pseudo-polynomial time.

Then, we derive polynomial-time bounds for several different valuation classes as direct corollaries of our main theorem, or via a slight modification of the algorithm. Namely, our base algorithm finds an EF1 orientation in polynomial time if every agent has a constant number of relevant items (Corollary~\ref{cor:relevant}), or when there exists a constant number of ``local'' item-types (Corollary~\ref{cor:local}). 
In addition, via straightforward modifications of the base algorithm, we can efficiently compute EF1 orientations for identical valuations (Theorem~\ref{thm: EF1-identicalvals}), or when the relevant items of the agents form {\em laminar} sets (Corollary~\ref{cor:laminar}). 

For EFX orientations, we begin by showing two strong negative results. Firstly, we show that it is \NPc to decide whether an EFX orientation exists even on graphs with vertex cover of size $10$, even when the valuations are additive and symmetric (Theorem~\ref{thm: EFX-NP_hard}). This result rules out the possibility of fixed-parameter algorithms for a large number of graph parameters. 
Furthermore, we show that if we consider multigraphs instead of graphs, i.e., we allow parallel edges, finding an EFX allocation is \NPh even when we have $8$ agents with symmetric and additive valuations (Theorem~\ref{thm: EFX-multigraph}). 
We complement these intractability results with a fixed parameter algorithm, for which the analysis is rather involved, parameterized by the slim tree-cut width of the underlying graph; this is essentially the best result someone could hope for, given our previous results.

\ifshort
\smallskip
\noindent{\emph{Due to space constraints, some details, marked with $\star$, are omitted and are available in the supplementary material.}}
\fi

\subsection{Related Work}

The recent survey by \citeauthor{AmanatidisABFLMVW2023}~\shortcite{AmanatidisABFLMVW2023} provides a comprehensive coverage of work on fair division of indivisible goods.
In the section, we direct the reader to some other papers in particular that study EFX or EF1 and restrict the instance in different ways. 


As aforementioned, the question of whether EFX always exists is a well-known open question in Fair Division, currently, we have that \citeauthor{PlautR20}~\shortcite{PlautR20} prove EFX always exists for $2$ agents. For $3$ agents already this question is much harder, \citeauthor{chaudhury2024efx}~\shortcite{chaudhury2024efx} prove that EFX exists for $3$ agents but with additive valuations, and recently \citeauthor{akrami2022efx}~\shortcite{akrami2022efx} generalize this result so that only $1$ of these agents requires additive valuations (and the other $2$ agents may have arbitrary valuations).  
The paper by \citeauthor{goldberg2023frontier}~\shortcite{goldberg2023frontier} studies the intractability of EFX with just two agents. They find that even with a small instance like this, it quickly becomes intractable as the valuations become more general -- namely computing an EFX allocation for two identical agents with submodular valuations is \PLS-hard. However, they propose an intuitive greedy algorithm for EFX allocations for weakly well-layered valuations; a class of valuations which they introduce. 
An example of relaxation of EFX that has been studied is EFkX, envy freeness up to $k$ goods, \citeauthor{ijcai2022p3}~\shortcite{ijcai2022p3} study EF2X and prove existence for agents with additive valuations (and some other minor restrictions).
A recent paper by \citeauthor{zhou2024EFX}~\shortcite{zhou2024EFX} studies EFX allocations in the mixed setting on graphs, where agents only have valuations for edges adjacent to them and these can be positive or negative. 
They treat orientations as a special case of their problem and show that deciding if an EFX orientation exists is \NPc.
The paper by \citeauthor{payan2023relaxations}~\shortcite{payan2023relaxations} also studies graph restrictions but these are subtly different to that of \cite{christodoulou2023fair}. In this model, edges are not items but instead, they are where EFX (/other fairness notions) must apply, intuitively this aims to capture a model where we want envy freeness between an agent and some of their neighbors. 
Some studies look at EFX where we (may) decide to leave some items unallocated. We refer to this as EFX with charity, \citep{caragiannis2019envy,chaudhury2021little}, where not all items are allocated and these leftover items are said to be ``given away to charity''. 
Moreover, \cite{CaragiannisKMPSW19} introduce EFX0 and \cite{kyropoulou2020almost} study this.
An allocation satisfies EFX0 if for one agent $i$ they are not envious of any other agent $j$'s bundle after removing any item for which agent $i$ doesn't have a positive value, i.e. EFX but we exclude items which have no value/zero value to agent $i$. 
Moreover, another model which may be of interest is that of connected fair division - originally introduced by \cite{bouveret2017fair} and more recently \cite{DeligkasEGHO2021} study this under the lens of parameterized complexity - in which there are some items which cannot be separated i.e. have some \textit{connectivity} constraints.

\section{Preliminaries}
\label{sec:prelims}
Throughout the paper we consider 
 $\items=\{a_1, a_2, \ldots, a_m\}$ to be a set of indivisible items and $N=\{1, 2, \ldots, n\}$ to be a set of agents. 
An allocation $\pi = (\pi_1, \pi_2, \ldots, \pi_n)$ is a partition of the items into $n$ (possibly empty) sets which we refer to as bundles. Thus, $\pi_i \cap \pi_j=\emptyset$ for every $i \neq j$ and $\bigcup_{i\in N}\pi_i = \items$, where the bundle $\pi_i$ is allocated to agent $i$. For an item $a \in \items$, we denote by $\pi(a)$ the agent who receives item $a$ in the allocation $\pi$. We refer to an allocation of a subset of items as a partial allocation.

Every agent $i \in N$ has a {\em valuation} function $\mathcal V_i$ that assigns a value $\mathcal V_i(S)$ for every subset $S \subseteq \items$, where $\mathcal V_i(\emptyset)=0$. 
$\mathcal V_i$ is {\em non-negative} if for all $S \subseteq N$ it holds $\mathcal V_i(S) \geq 0$; $\mathcal V_i$ is {\em monotone} if for all $S' \subset S$ it holds $\mathcal V_i(S') \leq \mathcal V_i(S)$; $\mathcal V_i$ is {\em additive} if there exist non-negative values $v_{i1}, v_{i2}, \ldots, v_{im}$ such that for every $S \subseteq \items$ it holds $\mathcal V_i(S) = \sum_{j \in S}v_{ij}$.


For an allocation $\pi$ and we say that agent~$i$ \emph{envies} agent $j$, or alternatively that there is \emph{envy from $i$ towards~$j$}, if $\mathcal V_i(\pi_j) > \mathcal V_i(\pi_i)$. 
An allocation is {\em fair} if envy can be eliminated in some particular way.

\iflong
\begin{definition}[EF]
An allocation $\pi$ is envy-free (EF), if for every pair of agents $i,j \in N$ it holds that $\mathcal V_i(\pi_i) \geq \mathcal V_i(\pi_j)$.
\end{definition}
\fi

\begin{definition}[EF1]
An allocation $\pi$ is envy-free up to one item (EF1), if for every pair of agents $i,j \in N$ there {\em exists} an item $a\in \pi_j$ such that $\mathcal V_i(\pi_i) \geq \mathcal V_i(\pi_j\setminus a)$.
\end{definition}

\begin{definition}[EFX]
An allocation $\pi$ is envy-free up to any item (EFX), if for every pair of agents $i,j \in N$ and {\em every} $a\in \pi_j$ it holds that $\mathcal V_i(\pi_i) \geq \mathcal V_i(\pi_j\setminus a)$.
\end{definition}

\paragraph{\bf Relevant items.}
We say that an item $a$ is \textit{relevant} for an agent $i$ if there is a set $S$ of items such that $\mathcal V_i(S\setminus {a}) < \mathcal V_i(S)$.  
For every agent $i$, we will denote the set of items relevant to $i$ by $A_i$. 
Similarly, for every item $a$, we let $N_a = \{i\in N\mid a\in A_i\}$ be the set of agents to which $a$ is relevant, we will also call $N_a$ the \emph{agent list} of $a$. 
We say that items $a$ and $b$ belong to the same \emph{group} if $N_a=N_b$.
Throughout the paper we assume that the union of all relevant sets is the set of items or, in other words, every item is relevant for at least one agent.

\paragraph{\bf Orientations.}
Using relevant items, we can define a subset of all possible allocations, that we call {\em orientations}. Formally, an allocation $\pi = (\pi_1, \pi_2, \ldots, \pi_n)$ is called an {\em orientation} if $\pi_i\subseteq A_i$ for all $i\in N$. In other words, in an orientation, the bundle of an agent contains {\em only} relevant items.
For $\phi \in \{ \text{EF1, EFX} \}$, we say that an allocation $\pi$ is a $\phi$ orientation if it is an orientation and in addition, it satisfies the corresponding fairness definition.


\section{EF1 orientations for monotone valuations}
\label{sec:EF1-orientations}

In this section we establish the existence of EF1 orientations when agents have monotone valuations via the construction of a pseudopolynomial-time algorithm and we identify several sub-classes of valuation functions for the agents where our algorithm, or a slight modification of it, finds an EF1 orientation in polynomial time.

\begin{algorithm}[b]
    \caption{EF1 orientations for monotone valuations}
    \label{alg:algorithm}
    \begin{algorithmic}[1] 
    \REQUIRE Set of items $A$, set of agents $N$, valuations $\mathcal{V}_i$, sets of relevant items $A_i$, $i\in N$, agent lists $N_a$, $a\in A$
    \ENSURE EF1 orientation $\pi$
        \STATE Let $\pi= (\pi_1,\ldots, \pi_n)$ such that $\pi_i=\emptyset$ for all $i\in N$.
        \STATE Let $G_{\pi} = (N, \emptyset)$ be a directed graph ( an ``envy-graph'' of $\pi$)
        \WHILE{$\exists$ item $a$ that is not in any $\pi_i$}
        \STATE Let $i\in N_a$ be an agent such that $i$ is a source-vertex in $G[N_a]$.
        \STATE Let $\pi_i = \pi_i\cup \{a\}$
        \FOR{$j\in N_a\setminus\{i\}$}
        \IF{$\mathcal V_j(\pi_j) < V_j(\pi_i)$}
        \STATE Add the edge from $j$ to $i$ if it does not exists.
        \ENDIF
        \ENDFOR
        \STATE Call Algorithm~\ref{alg:resolve_cycle} with $\pi$ and $G_{\pi}$ to eliminate all cycles in $G_{\pi}$.

        \ENDWHILE
        \STATE \textbf{return} $\pi$
    \end{algorithmic}
\end{algorithm}

\begin{algorithm}[h]
    \caption{Eliminating cycles in the envy-graph}
    \label{alg:resolve_cycle}
    \begin{algorithmic}[1] 
    \REQUIRE partial allocation $\pi$ with an envy-graph $G_{\pi}$
    \ENSURE partial allocation $\pi'$ with $\mathcal{V}_i(\pi'_i)\ge \mathcal{V}_i(\pi_i),i \in N$, such that $G_{\pi'}$ does not contain directed cycles
    \WHILE{ $\exists$ a directed cycle $C=(i_1, i_2, \ldots, i_{\ell})$ in $G_{\pi}$}
    \FORALL{$j\in [\ell-1]$}
    \STATE Let $\pi'_{i_j} = \pi_{i_{j+1}}\cap A_{i_j}$
    \ENDFOR
    \STATE $\pi'_{i_{\ell}} = \pi_{i_{0}}\cap S_{i_\ell}$
    \STATE Let $\pi = \pi'$
    \STATE Recompute $G_{\pi}$
    \ENDWHILE
\RETURN $\pi$
    \end{algorithmic}
\end{algorithm}

\begin{theorem}
\label{thm: EF1_exist}
 When agents have monotone valuations, an EF1 orientation always exists and can be computed in time $\Oh{mn^3r}$ where $r$ is the maximal range size of $\mathcal V_i$, $i\in N$.  
\end{theorem}
\begin{proof}
    A full pseudo-code of the algorithm is given in Algorithm~\ref{alg:algorithm}. The algorithm is inspired by the envy-cycle elimination algorithm by \citeauthor{LiptonMMS04}~\shortcite{LiptonMMS04}. Similarly to that algorithm, we are computing the allocation $\pi$ iteratively starting from an empty (partial) allocation keeping an envy graph $G_{\pi}$, which is a graph whose vertex set is precisely the set of agents $N$ and there is a directed edge from $i$ to $j$ if in the allocation $\pi$ the agent $i$ envies the agent $j$. In addition, we will be preserving that $\pi$ is an EF1 (partial) allocation such that $\pi_i\subseteq A_i$ for all $i\in N$. For the ease of the presentation of the proof, we will say that a pair of agents $i, j$ satisfy EF1-property if $\mathcal{V}_i(\pi_j) \le \mathcal{V}_i(\pi_i) $ or there exists $a\in \pi_j$ such that $\mathcal{V}_i(\pi_j\setminus \{a\}) \le \mathcal{V}_i(\pi_i)$.
    
    While the algorithm by \cite{LiptonMMS04} greedily picks any source-vertex $i$ in $G_{\pi}$ (that is a vertex without any edge pointing towards it; so no agent in $N$ envies agent $i$ its bundle in $\pi$) and allocates to $i$ an arbitrary item, this would not work for us, as the remaining items might not be relevant for the source vertices of $G_{\pi}$. Instead, we first pick an unassigned item $a$ that needs to be assigned and give it to an agent $i$ that is a source in the subgraph of $G_{\pi}$ induced by the agents in the set $N_a$ that are allowed to receive that item. Hence, the item $a$ is irrelevant for any agent $j$ that envies $i$ before allocating the item $a$. Therefore, these steps preserve both that $\pi_i\subseteq A_i$ for all $i\in N$ and that for all $i,j\in N$ the pair $i,j$ satisfies EF1-property.

    Furthermore, in order to always find a source-vertex in this induced subgraph of the envy graph, we also need to be able to eliminate cycles in $G_{\pi}$. An algorithm for eliminating all cycles in $G_{\pi}$ is described in Algorithm~\ref{alg:resolve_cycle}. Let $C = (i_1, i_2, \ldots, i_\ell)$ be a cycle such that the agent $i_j$ envies the bundle of the agent $i_{j+1}$ for $j\in [\ell-1]$ and the agent $i_\ell$ envies the bundle of the agent $i_1$. Similarly to \citeauthor{LiptonMMS04}, we shift the bundles in the opposite direction of the cycle, which is known to eliminate the envy on the cycle and does not create any new envy. However, after we execute this shift of bundles, the bundle of an agent $i$ on the cycle could contain some items that are not in their set $A_i$ of the relevant items. So we remove from the bundle of the agent $i$ any item that is not relevant for them. Note that this does not change the valuation of $i$ for its bundle. Let us denote by $\pi$ the allocation before the cycle-elimination step and by $\pi'$ the allocation after the cycle-elimination step. As we already discussed, for all $i\in N$ we have $\mathcal{V}_i(\pi'_i) \ge \mathcal{V}_i(\pi_i)$ (the value either increased or the bundle did not change). Now let $i,j\in N$.     
    We know that for all $j'\in N$, the pair $i,j'$ satisfies EF1-property in $\pi$. It follows that if $\pi'_j = \pi_j$, then $i,j$ satisfies EF1-property as well. Else $\pi'_j \subseteq \pi_{k}$ for some $k\in N$. It follows from monotonicity of the valuation that $\mathcal{V}_i(\pi'_j)\le \mathcal{V}_i(\pi_k)$ and $\mathcal{V}_i(\pi'_j\setminus a)\le \mathcal{V}_i(\pi_k\setminus a)$ for all $a\in \pi_k$ and so the pair $i,j$ also satisfies the EF1-property in $\pi'$. 
    
    It follows from the above discussion and from the fact that Algorithm~\ref{alg:algorithm} stops only when all items have been allocated that whenever Algorithm~\ref{alg:algorithm} stops it returns an EF1 orientation. It only remains to show that the algorithm 
    terminates. 

    Let us consider the vector $W^\pi = (\mathcal{V}_1(\pi_1), \mathcal{V}_2(\pi_2)), \ldots, \mathcal{V}_n(\pi_n))$. By definition, each coordinate $W^\pi_i$ of the vector $W^\pi$ has at most $r\le 2^{m}$ many possible values. In addition, in each cycle-elimination step all the coordinates corresponding to the agents on the cycle strictly increase and the remaining coordinates of $W^{\pi}$ do not change. Similarly, adding item $a$ to agent $i$ on line 5 of Algorithm~\ref{alg:algorithm} does not decrease any of the coordinates because of the monotonicity of the valuations. Since every coordinate can increase its value only $r-1$ times, it follows that there are less than $n\cdot r$ many cycle elimination steps in total, each can be executed in $\Oh{nm}$ time. Between any two cycle-eliminations we can only add less than $m$ items after each addition of an item, we need to update $G_{\pi}$ and check whether a cycle is created, which can be easily done in $\Oh{|V(G_\pi)| +|E(G_\pi)|} = \Oh{n^2}$ time.  Putting everything together, Algorithm~\ref{alg:algorithm} runs in $\Oh{nr\cdot(nm + mn^2)} = \Oh{mn^3r)} $ time.
\end{proof}



Note that if the number of relevant items for each agent is constant and bounded by some $\ell\in \mathbb{N}$, then the number of possible bundles for each agent and hence the range $r$ of its valuation function $\mathcal{V}_i$ is bounded by $2^\ell$. 
Therefore, Theorem~\ref{thm: EF1_exist} immediately implies the following corollary.
\begin{corollary}
\label{cor:relevant}
 For the monotone valuations with at most $l$ relevant items per agent, an EF1 orientation can be computed in time $\Oh{mn^3\cdot 2^l}$.
 \end{corollary}

\subsection{Local item-types}
While in Corollary~\ref{cor:relevant}, the number of relevant items per agent is constant, our algorithm provided in Theorem \ref{thm: EF1_exist} can be applied to compute an EF1 orientation in polynomial time in much more general settings. For instance, if the range of each valuation has size at most $m^d$ for some constant $d$, the running time is upper-bounded by $\Oh{m^{d+1}n^3}$.

One natural example of such valuations is as follows. Assume that each agent subdivides all items relevant to them into $d$ groups (which we will refer to as \emph{local item-types}) and only distinguishes items that are in different groups. In this case, the valuation each agent has for their bundle depends only on the number of received items from each local item-type. Since each of the $d$ groups contains at most $m$ items, there are at most $m^d$ possibilities for their value. 

\begin{corollary}
\label{cor:local}
 For the monotone valuations with at most $d$ local item-types per agent, an EF1 orientation can be computed in time $\Oh{m^{d+1}n^3}$.
 \end{corollary}

In general, when the range size of $\mathcal V_i$ is unbounded, we still can distinguish some settings for which the described algorithm (or its slight modification) is polynomial.

\subsection{Identical valuations}
While strictly speaking, identical valuations would mean that all items have to be relevant for all agents, we relax this notion slightly to better fit with the intended meaning of the relevant items as a restriction of the item an agent is allowed to receive.

\begin{definition} 
 A set of agents have {\em identical} valuations if there exists a function $\mathcal V$ 
 such that for every agent $i$ their valuation function is defined by $\mathcal V_i(B)=\mathcal V(B \cap A_i)$ for every $B \subseteq S$. 
\end{definition}
Similarly to the standard setting where agents with identical valuations cannot create envy cycles, we can show that the same is the case even with this relaxed definition of identical valuations and we obtain the following theorem.
\begin{theorem}
\label{thm: EF1-identicalvals}
An EF1 orientation can be computed in linear time when agents have identical monotone valuations.
\end{theorem}
\begin{proof}
We will show that it is not possible to have a cycle in the envy graph. That is that when we run Algorithm~\ref{alg:algorithm}, then whenever Algorithm~\ref{alg:resolve_cycle} is called as a subroutine, the condition on line 1 in Algorithm~\ref{alg:resolve_cycle} is always false and it returns the same partial allocation. 

For the sake of a contradiction, assume that $C= (i_1, i_2, \ldots, i_\ell)$ is a cycle in the envy-graph for some partial allocation $\pi$ such that $\pi_i\subseteq A_i$. That is the agent $i_j$ envies the agent $i_{j+1}$ for all $j\in [\ell-1]$ and the agent $i_\ell$ envies the agent $i_1$. Let $\mathcal{V}$ be the function such that for all $i\in N$ and for all $B\subseteq A$ we have $\mathcal{V_i}(B) = \mathcal{V}(B\cap A_i)$. Since, $\pi_i\subseteq A_i$ for all $i\in N$ it follows that $\mathcal{V}_{i}(\pi_{i}) = \mathcal{V}(\pi_{i})$ for all $i\in N$ and $\mathcal{V}_{i}(\pi_{j}) = \mathcal{V}(\pi_{j}\cap A_i)$ and by monotonicity $\mathcal{V}_{i}(\pi_{j}) \le \mathcal{V}(\pi_{j})$. It follows that 
\[ \mathcal{V}(\pi_{i_1}) < \mathcal{V}(\pi_{i_{j+1}}) < \cdots < \mathcal{V}(\pi_{i_\ell}) < \mathcal{V}(\pi_{i_1}), \]
which is a contradiction. Therefore, Algorithm~\ref{alg:algorithm} only has to assign every item $a$ once to a source-vertex in $G_\pi[N_a]$, without ever running Algorithm~\ref{alg:resolve_cycle} which takes linear time $\bigoh(mn)$.
\end{proof}

%

\subsection{Laminar agent lists}
The final setting for which we get a polynomial time algorithm is when items are arranged in some kind of hierarchical structure, where the most common items can be assigned to any agent and then we have more and more specialized items that only a smaller and smaller group of agents can get. One can think about it like this: agents need to undergo some training and which items you are allowed to receive depends on your specialization and on the amount of training you already received.
\begin{definition}
 We say that agent lists $N_a$, $a\in A$, are \textbf{laminar} if for any two items $a_1$ and $a_2$ it holds that either $N_{a_1}\cap N_{a_2}= \emptyset$ or one of the sets is a subset of another, i.e. $N_{a_1}\subseteq N_{a_2}$ or $N_{a_2}\subseteq N_{a_1}$ 
\end{definition}

\begin{corollary}
\label{cor:laminar}
For the monotone valuations with laminar agent lists, an EF1 orientation can be computed in time $\Oh{m^2n}$.
\end{corollary}
\begin{proof}
Recall that the algorithm described in the proof of Theorem \ref{thm: EF1_exist} picks undistributed items in random order. Here, instead of this, if the agent lists are laminar, we order the items in a tuple $(a_1, \dots, a_m)$ so that either $N_j\cap N_k =\emptyset$ or $N_k\subseteq N_j$ whenever $1\leq j < k \leq m$. In the algorithm, we will pick items exactly in this order. Consider the moment when a new item $a$ is allocated. Then shifting items along the cycle in $G_a$ will never result in allocating illegal items. Indeed, if some item $b$ was moved to agent $i\in N_a$, then $N_a \cap N_b \neq \emptyset$. Since $b$ was allocated before $a$, we conclude that $N_a \subseteq N_b$ and hence $i\in N_b$.

Therefore, each item is moved from not distributed to distributed precisely once, and potential cycle elimination afterwards takes time at most $\Oh{mn}$, so the EF1 orientation is computed in time at most $\Oh{m^2n}$.
\end{proof}


\section{EFX orientations}
\label{sec:EFX_orientations}

The paper by \citeauthor{christodoulou2023fair}~\shortcite{christodoulou2023fair} proves that it is \NP-complete to decide if an EFX orientation exists. We strengthen this result, by showing it is \NP-hard even when the graph has constant vertex cover. 

\begin{theorem}
\label{thm: EFX-NP_hard}
 Deciding whether an EFX orientation exists on graph $G$ is \NPh even when $G$ has a vertex cover of constant size.
\end{theorem}
\begin{proof}
We show a reduction from the \NP-complete problem \partition. An instance of partition consists of a multiset $S$ of positive integers. Let $T$ be the number of elements in $S$. Let $\sum_{t=1}^{T} S_t = 2B$. Given $S$ we want to decide if the elements can be divided into two subsets $S_1$ and $S_2$ such that the sum of elements in $S_1$ is equal to that of $S_2$. 
Given an instance of \partition, we will construct a graph $G=(V,E)$. We start by constructing a bipartite graph such that there is a vertex (i.e. an agent) for every element in $S$ on one side, call them $x_1, \cdots, x_T$, and two additional vertices $i$ and $j$ on the other.
Now we will create the edges, where $|E| = m$, the number of items and weights on edges are the valuation of that item for the agents on both endpoints. For all vertices $x_v \in \{x_1, x_2 ..., x_T\}$ we create an edge $(i, x_v)$ of weight $S_v$ and an edge $(j, x_v)$ of weight $S_v$. We create two copies of gadget $X$, we will call these $X_1$ and $X_2$. 
Gadget $X$ is a clique on $4$ vertices based on Example 1. in \cite{christodoulou2023fair} that has no EFX-orientation on its own.
Let the vertices in the gadget $X_h$ be $X_{h1}$ to $X_{h4}$ . Edges $(X_{h1},X_{h2})$ and $(X_{h3},X_{h4})$ have weight $B$. All other edges in $X_h$ have weight $1$.  
Finally, we create edges $(i,X_{11})$ and $(j,X_{21})$ of weight $B$.
This completes the construction, see Figure~\ref{fig:thm1}.
\begin{figure}[h]
\centering
\includegraphics[width=0.65\textwidth, page=1]{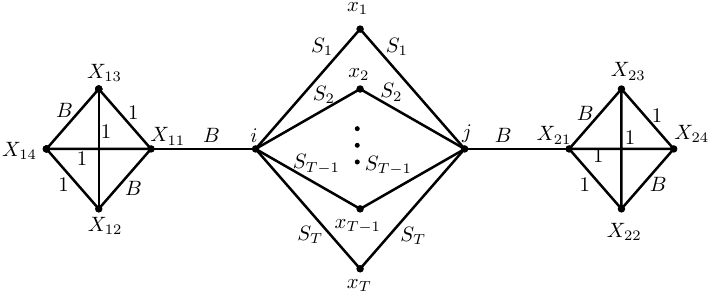}
\caption{The construction used in the proof of Theorem~\ref{thm: EFX-NP_hard}.}
\label{fig:thm1}
\end{figure}

\ifshort
One can now show that in order to get an EFX-orientation, the edges $(i, X_{11})$ and $(j, X_{21})$ have to be allocated to $X_{11}$ and $X_{21}$, respectively. In addition, the only way to satisfy all the remaining agents in the gadgets $X_1$ and $X_2$, both $X_{11}$ and $X_{21}$ need to receive additional edge of value $1$. Hence in order for both $i$ and $j$ to be satisfied, they each need to value their bundle at least $B$. Since to achieve this both $i$ and $j$ need to receive more than just one edge, each of agents $x_k$ for $k\in [T]$ has to receive at least one of its incident edges. Therefore, the only way for the graph $G$ to admit EFX-orientation is if the elements associated with the edges allocated to $i$ and the elements associated with the edges allocated to $j$ form a partition of $S$. 
\fi
\iflong 

Now, we will show that given a YES instance of \partition this graph construction will have an EFX orientation. Assume that we have a solution for \partition. Edges $(i,x_v)$ and $(j,x_v)$ corresponding to an element $v$ in $S_1$ will be orientated towards vertex $i$ and away from $j$ (i.e., $(i,x_v)$ is in the bundle of $i$ and $(j,x_v)$ is allocated to $x_v$). Analogously, Edges $(i,x_v)$ and $(j,x_v)$ corresponding to an element $v$ in $S_2$ will be orientated such that $(j,x_v)$ is in the bundle of $j$ and $(i, x_v)$ in the bundle of $x_v$. This allocation means that:
\begin{itemize}
    \item each of the agents/vertices $x_1, \ldots, x_T$ in the middle will only be allocated one item. The other item, relevant to them, will be of the same value as their current bundle. So they are not envious.
    \item agent $i$ will have a bundle with value $B$, because they are allocated exactly half of the value of $S$. They will not envy the middle vertices (who got the items which are relevant to them but they didn't get) because the sum of all the items they didn't get is $B$, so the bundle of each of the middle vertices has value at most $B$. 
    \item The same argument as above will hold for agent $j$.
\end{itemize}
It remains to show that there is an allocation of the remaining edges to the agents in the gadgets $X_1$ and $X_2$, such that the allocation is EFX, see \ref{fig:gadgetXh}. For $h\in \{1,2\}$, we allocate to 
\begin{itemize}
    \item $X_{h1}$ the edge $(X_{h4},X_{h1})$ and the edge between $X_{h1}$ and $i$ or $j$, respectively; 
    \item $X_{h2}$ the edges $(X_{h2}, X_{h4})$ and $(X_{h2}, X_{h1})$;
    \item $X_{h3}$ the edges $(X_{h1}, X_{h3})$ and $(X_{h2}, X_{h3})$;
    \item $X_{h4}$ the edge $(X_{h3}, X_{h4})$;
\end{itemize}
 It is easy to see that 
 the only agent which is envious at all is vertex $X_{h3}$ (but since we are looking for EFX, not an EF allocation, this is still satisfied.) 
Since $i$ ($j$) have already been allocated a bundle of value $B$, directing $(i,X_{11})$ ($(j,X_{21})$) towards $X_{11}$ ($X_{21}$), respectively, does not make the agents $i$ ($j$) envious. 

\begin{figure}[h]
\begin{center}
\includegraphics[width=0.25\textwidth]{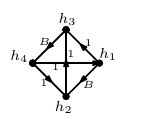}  

\end{center}
\caption{The Orientation solution for the Gadget $X_{h}$ used in Theorem \ref{thm: EFX-NP_hard} and Theorem \ref{thm: EFX-multigraph}.}
\label{fig:gadgetXh}
\end{figure}
For the other direction, assume that the graph $G$ has a valid EFX orientation. We know that for gadgets $X_1$ and $X_2$ to have an EFX orientation, the edges from $i$ and $j$ to the gadgets must be directed to vertex $X_{h1}$.
In addition, the bundle that contains the edge $(X_{h3}, X_{h4})$, cannot contain any other edge, since $B>2$, and the envy between $X_{h3}$ and $X_{h4}$ would not be resolved by removing any other edge from the bundle that contains $(X_{h3}, X_{h4})$. Hence, the bundle of $X_{h1}$ contains at least one more edge (either $(X_{h1}, X_{h3})$ or $(X_{h1}, X_{h3})$.
Therefore, for the agent $i$ not to be envious of $X_{11}$, they must receive a bundle with a value at least $B$. The same must hold for the agent $j$. Note for all $v\in [T]$, we have $S_v < B$, so both $i$ and $j$ will receive at least two items. It follows that for each $v\in [T]$, the agent $x_v$ has to receive one of its two incident edges. Since the set $S$ sums to $2B$, the only way to achieve this is to give $i$ a bundle of value exactly $B$ and a bundle of value exactly $B$ to agent $j$ such that these two bundles represent disjoint subsets of $S$. If such a split is possible, then there exists a valid \partition.
\fi
\end{proof}

\iflong 
\begin{theorem}
\fi
\ifshort 
\begin{theorem}[$\star$]
\fi
\label{thm: EFX-multigraph}
Deciding whether an EFX orientation exists on a multigraph $G$ is NP-hard even when $G$ has a constant number of vertices.
\end{theorem}

\ifshort
\begin{proof}[Proof Sketch]
We begin with a construction very similar to that of Theorem \ref{thm: EFX-NP_hard}. The only difference is that instead of creating a set of vertices $x_1 \cdots x_T$ we will create $T$ many edges between vertices $i$ and $j$, see Figure \ref{fig:thm2}.
\end{proof}
\fi

\iflong
\begin{proof}
    Similar to that of Theorem~\ref{thm: EFX-NP_hard} we will reduce from \partition. 
Given an instance of \partition, we will construct a graph $G=(V,E)$.  We create two copies of gadget $X$, called $X_1$ and $X_2$, the same as in \ref{thm: EFX-NP_hard}. 
Now we will create the edges, where $|E| = m$, the number of items. We create $T$ many edges between vertices $i$ and $j$, each associated with a distinct member of $S$. Finally, we create edges $(i,X_{11})$ and $(j,X_{21})$ of weight $B$.

Now, we will show that given a YES instance of \partition this graph construction will have an EFX orientation. Assume that we have a solution for \partition. We give the set of items in $S_1$ to agent $i$ and $S_2$ to agent $j$ (plus the same assignment of items to gadget $X$ exactly the same as in Theorem~\ref{thm: EFX-NP_hard}). It is easy to see that this is a valid EFX-orientation. This is because neither $i$ nor $j$ will envy one or another, as they have a bundle of the same value and for any other pair, the analysis is exactly the same as in Theorem~\ref{thm: EFX-NP_hard}. 
For the other direction, following the analogous discussion as in Theorem~\ref{thm: EFX-NP_hard}, we see that $i$ values the bundle of $X_{11}$ exactly $B$ and there is an item in the bundle of $X_{11}$ that $i$ values $0$. Therefore, $i$ needs to value its bundle at least $B$. Analogously, $j$ needs to value its bundle at least $B$. Hence, we can only have a valid EFX orientation if it is possible to assign the $T$ edges to $i$ and $j$ in an envy-free way, i.e. they both have a bundle of the value $B$.
\end{proof}
\fi

\begin{figure}[h]

\centering
\includegraphics[width=0.65\textwidth, page=2]{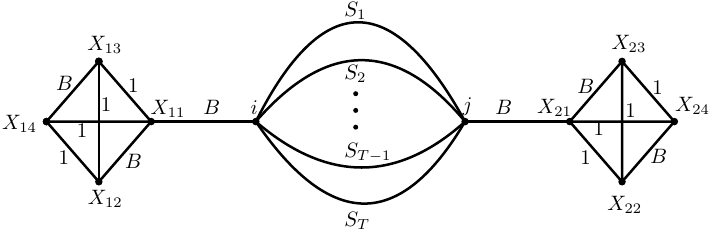}
\caption{The construction used in the proof of Theorem~\ref{thm: EFX-multigraph}.}
\label{fig:thm2}
\end{figure}


\subsection{Slim tree-cut width and the FPT algorithm}

Examining the hardness of \cite{christodoulou2023fair}, we can see that their reduction can be made to work on graphs with constant maximum degree. In combination with Theorem~\ref{thm: EFX-NP_hard}, this shows that efficient algorithms for deciding the existence of EFX orientations are unlikely already for very restricted settings. In this section, we show an efficient algorithm for a setting that is basically on the limit of tractability from the point of view of parameterized complexity (in a sense, that more general parameters studied so far would contain some of the hard instances). Before we can go into more details though, we have to introduce some notions and notations.

\paragraph{Parameterized Complexity.} 
An instance of a parameterized problem $Q\subseteq \Sigma\times\mathbb{N}$, where $\Sigma$ is fixed and finite alphabet, is a pair $(I,k)$, where $I$ is an input of the problem and~$k$ is a \emph{parameter}.
The ultimate goal of parameterized algorithmics is to confine the exponential explosion in the running time of an algorithm for some \NPh problem to the parameter 
and not to the instance size. The best possible outcome here is the so-called \emph{fixed-parameter} algorithm with running time $f(k)\cdot |I|^\Oh{1}$ for any computable function $f$. 
That is, for every fixed value of the parameter, we have a polynomial time algorithm where the degree of the polynomial is independent of the parameter. 
For a more comprehensive introduction to parameterized complexity, we refer \iflong the interested reader \fi to the monograph of~\citeauthor{CyganFKLMPPS2015}~\shortcite{CyganFKLMPPS2015}.

As follows from our reduction, deciding whether EFX orientation exists in the graph setting is hard, even when the underlying graph has constant vertex cover and tree-cut width. In particular, this rules out most of the vertex-separator based parameters. However, we show that the recently introduced parameter slim tree-cut width (also equivalent to super edge-cut width, see \cite{ganian2024slim}) allows us to achieve tractability.

For simplicity, here we work with super edge-cut width. For a graph $G$ and a spanning tree $T$ of $G$, let the \emph{local feedback edge set} at $v\in V$ be 
$E_{\loc}^{G,T}(v)=\{uw\in E(G)\setminus E(T)~|~$ the unique path between $u$ and $w$ in $T$ contains $v\}.$

\begin{definition}
The edge-cut width of the pair $(G,T)$ is $\widthshort(G,T)=1+\max_{v\in V} |E_{\loc}^{G,T}(v)|$, and the edge-cut width of $G$ $($denoted $\ecw(G))$ is the smallest \width\ among all possible spanning trees $T$ of $G$.
\end{definition}

If in the last definition, we allow to choose the spanning tree in any connected supergraph of $G$, this leads to the notion of \emph{\supwidth} (denoted by $\supwidthshort(G)$): \[\supwidthshort(G)=\min\{\widthshort(H, T)~|~H\supseteq G \text{, } T\text{ -- spanning tree of }H\}.\]

Super edge-cut width is a strictly more general parameter than degree+treewidth and feedback edge number, but it is more restrictive than tree-cut width \cite{ganian2024slim}.

\iflong 
\begin{theorem}
\fi
\ifshort 
\begin{theorem}[$\star$]
\fi
\label{thm: EFX_FPT}
  Deciding whether an EFX orientation exists for a given graph $G$ is fixed-parameter tractable if parameterized by the slim tree-cut width of $G$.
\end{theorem}

\subsection{Proof of Theorem~\ref{thm: EFX_FPT}}

Let $H$ be the supergraph of $G$ with the spanning tree $T$ such that 
$\supwidthshort(G)=\widthshort(H, T)=k$. 
We begin by recalling some basic properties of the super 
edge-cut width, in particular, we adapt the related notion of the boundary of a vertex, firstly introduced in \cite{GanianK2021} for a weaker parameter.

For $v \in V(T)$, let $T_v$ be the subtree of $T$ rooted at $v$, let $V_v=V(G \cap T_v)$, and let $\bar V_v=N_G(V_v)\cup V_v$. We define the \emph{boundary} $\delta(v)$ of $v$ to be the set of endpoints of all edges in $G$ with precisely one endpoint in $V_v$ (observe that the boundary can never have a size of $1$). 
A vertex $v$ of $T$ is called \emph{closed} if $|\delta(v)| \le 2$ and \emph{open} otherwise. 

\begin{observation}[\cite{GanianK2021}]
\label{obs:basiclfesproperties}
Let $v$ be a vertex of $T$. Then:
\begin{enumerate}
\item $\delta (w)=\{v,w\}$ for every closed child $w$ of $v$ in $T$, and $vw$ is the only edge between $V_w$ and $V\setminus V_w$ in $G$.
\item $|\delta(v)|\leq 2k+2$. 
\item Let $\{v_i|i\in [t]\}$ be the set of all open children of $v$ in $T$. Then, $\delta(v)\subseteq \cup _{i=1}^t \delta(v_i) \cup \{v\}\cup N_G(v)$ and $t\leq 2k$.
\end{enumerate}
\end{observation}

We can now define the records that will be used in our dynamic program. Intuitively, these records will be computed in a leaf-to-root fashion and will store at each vertex $v$ information about possible EFX orientations for the subtree of $T$ rooted at $v$.

Let $R$ be a binary relation on $\delta(v)$, and $S$ a subset of $\delta(v)$. 

\begin{definition}
    \label{def:record}
We say that $(R,S)$ is a \emph{record} at vertex $v$ if there exists an orientation $D$ of $G$ restricted to vertices in $\bar V_v$ and edges with at least one endpoint in $V_v$ such that:
\begin{enumerate}
    \item The partial allocation defined by $D$ is EFX between any two vertices in $V_v$.  
    \item $R$ is the set of all arcs of $D$ that have precisely one endpoint in $V_v$.
    \item For every $w\in \delta (v)$, if there is envy from some vertex $u\in V_v$ towards $w$ in $D$ then $w\in S$.
    \item The in-degree of every $w\in S$ in $D$ is equal to one.
\end{enumerate}
\end{definition}
 We call such a digraph $D$ a \emph{partial solution} at $v$. We say that $D$ is a \emph{witness} of the record $(R,S)$. Denote the set of all records for $v$ by $\mathcal R (v)$, then $|\mathcal{R}(v)|\leq 2^{\bigoh(k^2)}$.

 Intuitively, the set $S$ in a record is intended to capture all the vertices of $\delta(v)$ towards which there can be envy in the resulting solution. This is why we require all the vertices in $S$ to have the in-degree of one. Otherwise, there would be a possibility to remove items without changing the envy: note that the only relevant item for both agents is their shared edge. 

\iflong
In our algorithm, we will obtain records at each vertex by combining records of its children. For this, we need to ensure that the choices of $S$ in the records of siblings are consistent with respect to the shared vertices. We achieve this by allowing some freedom in the choice of $S$ for a given partial solution $D$: some of the vertices can be added to $S$ even if there is no envy towards them from the vertices of $T_v$, just to ensure that they have the in-degree of one in $D$ and hence can be envied by vertices from other subtrees in the combined solution. 
\fi

Note that if $v_i$ is a closed child of $v$, then by Observation \ref {obs:basiclfesproperties}, $\mathcal R(v_i)$ can contain only the records $(\{ v_iv\}, \emptyset)$, $(\{v_iv\}, \{v\})$, $(\{vv_i\}, \emptyset)$ and $(\{vv_i\}, \{v_i\})$.
The root $r$ of $T$ contains at most one record $(\emptyset, \emptyset)$, which happens if and only if the instance is a YES instance (otherwise $\mathcal R (r)=\emptyset$).

\iflong
\begin{observation}
Let $v_i$ be a closed child of $v$. Then:
\begin{itemize}
\item if $\mathcal R(v_i)$ contains $(\{ v_iv\}, \emptyset)$, then it also contains $(\{ v_iv\}, \{v\})$,
\item If $\mathcal R(v_i)$ contains $(\{ vv_i\}, \{v_i\})$, then it also contains $(\{vv_i\}, \emptyset)$.
\end{itemize}
\end{observation}
\begin{proof}
 Assume that $(\{ v_iv\}, \emptyset) \in \mathcal R(v_i)$, and let $D_i$ be any witness of $(\{ v_iv\}, \emptyset)$ at $v_i$.
 Since the in-degree of $v$ in $D_i$ is equal to its one and only in-neighbor in $D_i$ is $v_i \in \delta(v_i)$, we conclude that $D_i$ is also a witness of $(\{ v_iv\}, \{v\})$ at $v_i$.

 For the second implication, let $D_i$ be any witness of $(\{ vv_i\}, \{v_i\})$ in $v_i$. Then the only in-neighbor of $v_i$ in $D_i$ is $v$, and hence there is no envy from any $u \in V_{v_i}$ towards $v_i$.
 In particular, $D_i$ is also a witness of $(\{ v_iv\}, \emptyset)$. 
\end{proof}

\fi
\ifshort
\begin{observation}[$\star$]
Let $v_i$ be a closed child of $v$. Then:
\begin{itemize}
\item if $(\{ v_iv\}, \emptyset) \in \mathcal R(v_i)$, then  $(\{ v_iv\}, \{v\}) \in \mathcal R(v_i)$,
\item If $(\{ vv_i\}, \{v_i\}) \in \mathcal R(v_i)$, then $(\{vv_i\}, \emptyset)\in \mathcal R(v_i)$.
\end{itemize}
\end{observation}
\fi

\iflong
\begin{lemma}
\fi
\ifshort
\begin{lemma}[$\star$]
\fi
\label{lem:combRecordsnew} 
Let $v\in V(G)$ have $c>0$ children in $T$, and assume we have computed $\mathcal{R}(v_i)$ for each child $v_i$ of $v$. Then $\mathcal{R}(v)$ can be computed in time at most $2^{\bigoh (k^3)} \cdot c$.
\end{lemma}
\iflong
\begin{proof}
\fi
\ifshort
\begin{proof}[Proof sketch]
\fi
 Without loss of generality, let the open children of $v\in V(G)$ be $v_1,\dots,v_t$, then $t\leq 2k$ by Point 3 of Observation~\ref{obs:basiclfesproperties}. Let $C$ be the set of remaining (closed) children of $v$, i.e. $C=\{v_{t+1},\dots,v_c\}$. Denote $V_i=V_{v_i}$, $\bar V_i=\bar V_{v_i}$, $i\in[c]$, and $V_0= \delta(v) \setminus \cup_{i=1}^c V_i$.
 Note that all $V_i$, $i\in [c]_0$, are pairwise disjoint. 
 If the record set of some child is empty, we conclude that $\mathcal{R}(v)=\emptyset$. Otherwise, we branch over all choices of $(R_i, S_i) \in \mathcal{R}(v_i)$ for each individual open child $v_i$ of $v$. We also branch over all orientations $R_0$ of edges $\{uv|u\in V_0\}$. 
 
 Let $R'=\bigcup_{j\in [t]_0} R_j$. We branch over subsets $S_0$  of $V_0 \setminus \{v\}$ with precisely one incoming arc in $R'$. Let $S'=\bigcup_{j\in [t]_0} S_j$, 
 if $R'$ is not anti-symmetric or contains two different arcs $u_1w$ and $u_2w$ for some $w\in S'$, we discard this branch. We also discard it if $V_i \cap S' \neq V_i \cap S_i$ for some $i \in [t]$. 
Otherwise, we create a trial record $(R,S)$, where $R$ is the restriction of $R'$ to those arcs which have precisely one endpoint in $V_v$ and $S=S' \cap \delta (v)$.
\iflong 
Intuitively, to check whether $(R,S)$ is a record at $v$, we will try to choose records for the closed children such that their partial solutions agree with $(R_i, S_i)$, $i\in [t]_0$. 
As the number of closed children is not bounded by a parameter, we can not branch over all combinations of their records. Instead of this, we branch over the following cases to model all possible in-neighborhoods of $v$ in the partial solution. 
\fi
\ifshort
We branch over 4 options for the in-neighbors of $v$ in the partial solution, and in each case try to choose suitable records for the closed children. 
\fi

 \textbf {Case 0: no in-neighbors.} If $v$ has no in-neighbors in $R'$, and every closed child $v_i$ contains the record $(\{vv_i\}, \{v_i\})$ (or $(\{vv_i\}, \emptyset)$ if $V_v(\{v_iv\})=0$ ), and every $w\in N_G(v) \setminus C$ with $\mathcal V_v(\{wv\})>0$ is in $S'$, we add the record $(R,S)$. 

 \textbf {Case 1: one open in-neighbor.} If there is precisely one incoming arc $uv$ to $v$ in $R'$, and for every closed child $v_i$ of $v$, $\mathcal R(v_i)$ contains the record $(\{vv_i\},\emptyset)$, let $s_1=\mathcal V_v(\{uv\})$. If $s_1<\mathcal V_v(\{v_iv\})$ for some $v_i\in C$, and $\mathcal R(v_i)$ does not contain the record $(\{vv_i\},\{v_i\})$, we discard this case\iflong, since this would mean that $v$ envies $v_i$  while $v_i$ receives more than one edge in any partial solution\fi. We also discard it if $s_1<\mathcal V_v(\{wv\})$ for some $w\in N_G(v) \setminus (C\cup S')$. 
Otherwise, we add the records $(R,S)$ and $(R,S \cup \{v\})$.

 \textbf {Case 1': one closed in-neighbor.} 
If $R'$ has no incoming arcs $uv$ to $v$, we model partial solutions where the unique in-neighbor of $v$ is one of its closed children. We branch over the choices of this closed child $v_i$. For every fixed $v_i \in C$ such that $\mathcal R(v_i)$ contains the record $(\{v_iv\}, \{v\})$, we compare $s_i=V_v(\{v_iv\})$ with all the values $s_j$ over the rest of the closed children $v_j$ of $v$. If $s_j>s_i$ for some $j$ and $\mathcal R(v_j)$ does not contain the record $(\{vv_j\},\{v_j\})$, we discard the case. We also discard it if for some $j \neq i$, $\mathcal R(v_j)$ does not contain $(\{vv_j\},\emptyset)$. Otherwise, we compare $s_i$ with $\mathcal V_v(\{wv\})$ for every $w\in N_G(v) \setminus C$. If for some $w \not \in S'$ the latter is larger, we discard the branch as well. Otherwise, we add the record $(R,S  \cup\{v\})$. If $\mathcal R(v_i)$ contains the record $(\{v_iv\}, \emptyset)$, we additionaly add $(R,S)$. 

\textbf {Case 2: two or more in-neighbors.} Finally, if $v \not \in S'$, we model the subcase when $v$ has more than one in-neighbor in a partial solution. For every closed child 
$v_i$ of $v$ such that $(\{v_iv\}, \emptyset) \not \in \mathcal{R}(v_i)$, $vv_i$ must be oriented towards $v_i$. On the other hand, for every $v_i \in C$ such that $\mathcal{R}(v_i)$ contains the record $(\{v_iv\}, \emptyset)$, by monotonicity of $\mathcal V_v$  we can just greedily orient the edge $vv_i$ towards $v$. The corresponding value of $v$ will be $s=\mathcal V_v(\{wv|wv \in R'\} \cup \{v_iv|v_i\in C, (\{v_iv\}, \emptyset)\in \mathcal{R}(v_i)\}$. If there is closed child $v_i$ such that $s<V_v(\{v_iv\})$ and $\mathcal R(v_i)$ does not contain the record $(\{vv_i\},\{v_i\})$, we discard the branch. Moreover, if there is closed child $v_i$ such that $(\{v_iv\}, \emptyset) \not \in \mathcal{R}(v_i)$ and $(\{vv_i\},\emptyset) \not \in \mathcal R(v_i)$, we discard the branch. We also discard it if $s<V_v(\{vw\})$ for some $w \in N_G(v)\setminus C$ such that $w\not \in S'$. Otherwise, we add the record $(R,S)$. 

\iflong
For the running time, recall that in order to construct $\mathcal R (v)$ the algorithm branched over the choice of at most $(2k+1)$ many binary relations $R_i$ on the boundaries of open children and $v$ itself. According to Observation~\ref{obs:basiclfesproperties}.2, there are at most $\bigoh(2^{(2k+2)^2)})$ options for every such relation, which dominates the number of possible choices of subsets $S_i$ of the boundaries.
Therefore, we have at most $\bigoh((2^{(2k+2)^2})^{2k+1}  \le 2^{\bigoh(k^3)}$ branches. In particular, this dominates the time required to check, for a fixed $(R', S')$, whether the brunch should be kept and proceed to the closed children. In the latter case, we need at most linear in $c$ time to compute the values $s_1$ and $s_2$ and check all potential envies between $v$ and its closed children. Hence, $\mathcal R(v)$ can be computed in time  $2^{\bigoh (k^3)}\cdot c$.
\fi

\ifshort
For the running time, note that we branch over the choice of at most $2k+1$ many binary relations $R_i$, and there are at most $\bigoh(2^{(2k+2)^2)})$ options for every such relation, which dominates the number of possible choices of subsets $S_i$ of the boundaries.
Therefore, we have at most $\bigoh((2^{(2k+2)^2})^{2k+1}  \le 2^{\bigoh(k^3)}$ branches. In each branch, we need at most linear in $c$ time to traverse closed children of $v$. Hence, $\mathcal R(v)$ can be computed in time  $2^{\bigoh (k^3)}\cdot c$.
\fi

To establish the correctness of the procedure described above, which we will refer to as CRC($v$) (combining records of children of $v$), we prove the following two claims:

\iflong 
\begin{claim}
\fi
\ifshort 
\begin{claim}[$\star$]
\fi
\label{cl:combRecords1}
If $(R^*,S^*)\in \mathcal R(v)$, then CRC($v$) adds it. 
\end{claim}
\iflong
\begin{proof}[Proof of the Claim]
Assume that $(R^*,S^*)$ is a record at $v$, and let $D$ be its witness. Let $D_0$ be the restriction of $D$ to the vertex set $V_0$ and to those arcs which are incident to $v$. For each $i \in [c]$, let $D_i$ be the restriction of $D$ to the vertex set $\bar V_i$ and to those arcs that have at least one endpoint in $V_i$. Let $R_i$ be the set of arcs of $D_i$ with precisely one endpoint in $V_i$. Let $S_i$ be the set of vertices $w \in \delta(v_i)$ such that the in-degree of $w$ in $D_i$ is equal to one and either $w \in S^*$ or some $u \in V_v$ envies $w$ in $D$. 
We claim that for each $i \in [c]$, $(R_i, S_i)$ is a record at $v_i$ and $D_i$ is a partial solution witnessing it. 

Indeed, Point. 2 and Point. 4 of Definition \ref{def:record} hold by construction. Since $D$ is a partial solution at $v$ and $D_i$ contains all the arcs of $D$ incident to $V_i$, the allocation defined by $D_i$ is EFX between any two vertices in $V_i$, so Point. 1 holds as well. Towards Point. 3, note that if $w\in \delta (v_i)$ and $u \in V_v$ envies $w$ in $D$, then since $D$ is a partial solution, the in-degree of $w$ in $D$ is equal to one. Hence, $S_i$ contains $w$ by construction. Therefore, $(R_i, S_i)$ is a record at $v_i$ and $D_i$ is a partial solution witnessing it.

By the assumptions of Lemma \ref{lem:combRecordsnew}, there will be some branch in CRC($v$)where:
\begin{itemize}
    \item the records $(R_i, S_i)$ for each $i \in [t]$ are chosen,
    \item $R_0$ is the arc set of $D_0$,
    \item $R'=\bigcup_{j\in [t]_0} R_j$,
    \item $S_0=(V_0 \setminus \{v\})\cap S^*$ (note that all these vertices have precisely one incoming arc in $D$, and it belongs to $R'$),
    \item $S'=\bigcup_{j\in [t]_0} S_j$.
\end{itemize}
Here $R'$ is a subset of arcs of $D$, so it is anti-symmetric. We will ensure that $R'$ contains at most one incoming arc to each $w \in S'$, by showing that the in-degree of every such $w$ in $D$ is equal to one. If $w \in S^*$, this holds by definition of the partial solution. If $w \in S' \setminus S^*$, then $w$ belongs to $S_i$ for some $i \in [t]$. By construction of $S_i$, there is envy from some $u \in V_v$ towards $w$. Since $w \not \in S^*$, we conclude that $w \not \in \delta(v)$, in particular, $w \in V_v$. Since the allocation defined by $D$ is EFX if restricted to $V_v$, the in-degree of $w$ in $D$ is equal to one in this case as well.

For every $i\in [t]$, we will show that $V_i \cap S' = V_i \cap S_i$. Consider any $w\in V_i \cap S'$. As $S_0 \subseteq V_0$ and $V_0 \cap V_i =\emptyset$, we conclude that $w \in S_j$ for some $j \in [t]$, then its in-degree in $D_j$ is equal to one. Moreover, either $w \in S^*$ or some $u\in V_v$ envies $w$ in $D$. In either case, the in-degree of $w$ in $D$ (and hence in $D_i$, as $w \in V_i$) is equal to one. Hence, $w \in S_i$. 

Therefore, this branch of CRC($v$)is not discarded and we obtain $R$ and $S$ by restricting $R'$ to the arcs with precisely one endpoint in $V_v$ and $S'$ to $\delta (v)$, correspondingly. By construction, $R'$ contains all the arcs of $D$ with precisely one endpoint in $V_v$. Since $R$ is obtained by restricting $R'$ to those arcs and $D$ is a witness of $(R^*, S^*)$, we have that $R^* =R$. 

\begin{observation}
\label{obs: eqSSstar}
$S^* \in \{S,  S \cup \{v\}\}$.
\end{observation}
\begin{proof}
Consider any vertex $w\in S^* \setminus \{v\}$, then the in-degree of $w$ in $D$ is equal to one. If $w\in V_0 \setminus \{v\}$, by the choice of $S_0$ we have that $w \in S_0 \subseteq S'$. If $w \in V_i$ for some $i \in [t]$, the unique in-neighbor of $w$ belongs to $D_i$ and hence $w \in S_i \subseteq S'$. As $w \in S^* \subseteq \delta (v)$, we conclude that $w \in S' \cap \delta (v) =S$. Therefore,
$S^* \subseteq S \cup \{v\}$.

Conversely, consider any vertex $w\in S$. If $w\in S_0$, it belongs to $S^*$. If $w\in S_i$ for some $i \in[t]$, 
it follows from the construction of $S_i$ that either $w\in S^*$ or some $u\in V_v$ envies $w$ in $D$. In the latter case $w\in S^*$ as well, since $w \in S \subseteq \delta(v)$ and $D$ is a witness of $(R^*, S^*)$. Hence, $S \subseteq S^*$. 
\end{proof} 

\begin{observation}
\label {obs: v_envy_S'}
If $v$ envies some $w \not \in C$ in $D$, then $w\in S'$.
\end{observation}
\begin{proof}
If $w \in V_0$, then $w \in \delta(v)$. Since $D$ is partial solution witnessing $(R^*, S^*)$, $w$ belongs to $S^*$. By Observation \ref{obs: eqSSstar} we have that $S^* \subseteq S \cup \{v\}$, so $w \in S \subseteq S'$.

Otherwise $w=v_i$ for some $i\in[t]$. As $v,w \in V_v$ and $D$ define an EFX allocation between vertices of $V_v$, the in-degree of $w$ in $D$ (and hence in $D_i$) is equal to one. Since $w=v_i \in \delta (v_i)$ and $v\in V_v$ envies $w$ in $D$, $S_i$ contains $w$ by construction. Therefore, every such $w$ belongs to $S'$. 
\end{proof} 

Based on the in-neighborhood of $v$ in $D$, we distinguish the following cases. If $v$ has no in-neighbors in $D$, we claim that $(R^*, S^*)$ was added to $\mathcal R(v)$ in \textbf{Case 0} of CRC($v$). Indeed, as the in-degree of $v$ is equal to zero, $v \not \in S^*$ and hence $S=S^*$. For every closed child $v_i$ with $V_v(\{v_iv\})>0$, $v$ envies $v_i$ in $D$ and hence the in-degree of $v_i$ in both $D$ and $D_i$ is equal to one. Therefore, $D_i$ witnesses that $(\{vv_i\}, \{v_i\})\in \mathcal R(v_i)$.
For every other closed child $v_i$, $D_i$ witnesses that $(\{vv_i\}, \emptyset)\in \mathcal R(v_i)$.
Moreover, for every $w \in N_G(v) \setminus C$ with $V_v(\{wv\})>0$, we have that $v$ envies $w$ in $D$, so $w \in S'$ by Observation \ref{obs: v_envy_S'}. Hence, CRC($v$)adds $(R,S)=(R^*, S^*)$ to $\mathcal R(v)$.

The case when $v$ has a unique in-neighbor $u$ in $D$ and $uv \in R'$ is captured by \textbf{Case 1} of CRC($v$). Indeed, in this case for every closed child $v_i$ of $v$, $D_i$ has the arc $vv_i$, so it witnesses that $\mathcal R(v_i)$ contains the record $(\{vv_i\},\emptyset)$. Moreover, if for some $v_i\in C$, $\mathcal V_v(\{v_iv\})$ is larger than $s_1=\mathcal V_v(\{uv\})$, then the in-degree of $v_i$ in $D_i$ must be equal to one, so $D_i$ witnesses that $(\{vv_i\},\{v_i\}) \in \mathcal R(v_i)$.
Furthermore, if $s_1<\mathcal V_v(\{wv\})$ for some $w\in N_G(v) \setminus C$, then $v$ envies $w$ in $D$, so $w$ belongs to $S'$ by Observation \ref{obs: v_envy_S'}. 
Hence, CRC($v$)adds $(R,S)$ and $(R,S \cup \{v\})$ to $\mathcal R(v)$, and one of these records is $(R^*, S^*)$. 

The next case we consider is when the unique in-neighbor of $v$ in $D$ is its closed child $v_i$, in particular $R'$ has no incoming arcs $uv$ to $v$.
We show that this situation is covered by \textbf{Case 1'} of CRC($v$), namely by the branch choosing $v_i$.
Indeed, $D_i$ witnesses that $(\{v_iv\}, \{v\}) \in \mathcal R(v_i)$. If $s_i=V_v(\{v_iv\})<V_v(\{v_jv\})$ for some other closed child $v_j$ of $v$, then the in-degree of $v_j$ in $D$ is equal to one, in particular, $D_j$ witnesses that $(\{vv_j\},\{v_j\}) \in \mathcal R(v_j)$. Moreover, for each $j \neq i$, $D_j$ witnesses that $(\{vv_j\},\emptyset) \in \mathcal R(v_j)$. If $s_i<\mathcal V_v(\{wv\})$ for some $w\in N_G(v) \setminus C$, then $v$ envies $w$ in $D$, so $w$ belongs to $S'$ by Observation \ref{obs: v_envy_S'}. Hence, CRC($v$)adds $(R,S  \cup\{v\})$ to $\mathcal R(v)$. If $S^*=S\cup \{v\}$, this is precisely the record $(R^*, S^*)$. Otherwise $S^*=S$ and $v\not \in S^*$. Then $v_i$ does not envy $v$ in $D$ (and hence in $D_i$), so $D_i$ witnesses that   
$(\{v_iv\}, \emptyset) \in \mathcal R(v_i)$. Therefore, CRC($v$) also adds $(R,S)=(R^*, S^*)$ to $\mathcal R(v)$. 

Finally, we show that the case when $v$ has more than one in-neighbor in $D$ is covered by \textbf{Case 2} of CRC($v$). Indeed, in this case, $v\not \in S^*$ and no vertices of $V_v$ envy $v$ in $D$. In particular, $v \not \in S_i$ for any $i\in [t]_0$, so $v \not \in S'$. 
For every closed child 
$v_i$ of $v$ such that $(\{v_iv\}, \emptyset) \not \in \mathcal{R}(v_i)$, we conclude that the edge $vv_i$ is oriented towards $v_i$ in $D_i$, otherwise the envy from $v_i$ towards $v$ would appear. In particular, $D_i$ witnesses that $(\{vv_i\}, \emptyset) \in \mathcal{R}(v_i)$. Moreover, by monotonicity of $\mathcal V_v$, the value of $v$ in $D$ is upper-bounded by $s=\mathcal V_v(\{wv|wv \in R'\} \cup \{v_iv|v_i\in C, (\{v_iv\}, \emptyset)\in \mathcal{R}(v_i)\}$. Since $D$ is a partial solution, for every closed child $v_i$ such that $s<V_v(\{v_iv\})$ the in-degree of $v_i$ in $D$ (and hence in $D_i$) is equal to one. Therefore, $D_i$ witnesses that $(\{vv_i\},\{v_i\}) \in \mathcal R(v_i)$. Moreover, if $s<V_v(\{vw\})$ for some $w \in N_G(v) \setminus C$, then $w\in S'$ by Observation. \ref{obs: v_envy_S'}. Hence, CRC($v$) adds $(R,S)=(R^*,S^*)$ to the record set of $v$. 
\end{proof}
\fi

\iflong 
\begin{claim}
\fi
\ifshort 
\begin{claim}[$\star$]
\fi
\label{cl:combRecords2}
If CRC($v$) adds $(R,S)$, then $(R,S) \in \mathcal R(v)$. 
\end{claim} 
\iflong 
\begin{proof}[Proof of the Claim]
Fix the branch where $(R,S)$ was added to $\mathcal R(v)$, we will construct a partial solution $D$ witnessing $(R,S)$ by gluing together partial solutions of children of $v$.  
For every $i\in[t]$, pick any witness $D_i$ of the record $(R_i, S_i)$ at $v_i$. Denote by $D_0$ the digraph on $V_0$ with the arc set $R_0$. For the closed children of $v$, we choose partial solutions as follows.

If $(R,S)$ was added in \textbf{Case 0}, for each closed child $v_i$ of $v$ we define $D_i$ to be a partial solution at $v_i$ witnessing the record $(\{vv_i\},\{v_i\})$, or $(\{vv_i\}, \emptyset)$ in case $V_v(\{v_iv\})=0$.

If $(R,S)$ was added in \textbf{Case 1}, for every closed child $v_i$ of $v$ with $\mathcal V_v(\{v_iv\})>s_1$, we pick a partial solution $D_i$ witnessing the record $(\{vv_i\},\{v_i\})$. For the rest of the closed children $v_i$, $\mathcal R(v_i)$ contains the record $(\{vv_i\},\emptyset)$ at $v_i$, let $D_i$ be the witness of this record. 

If $(R,S)$ was added in \textbf{Case 1'}, in the branch corresponding to $v_i \in C$, we define $D_i$ to be a partial solution at $v_i$ witnessing the record $(\{v_iv\},\emptyset)$, if there is such a record, or $(\{v_iv\},\{v\})$ otherwise.
For all $v_j \in C$ such that $\mathcal V_v(\{v_jv\})>\mathcal V_v(\{v_iv\})$, let $D_j$ be a partial solution at $v_j$ witnessing its record $(\{vv_j\},\{v_j\})$.
For the rest of closed children $v_j \in C$, let $D_j$ be a partial solution at $v_j$ witnessing its record $(\{vv_j\},\emptyset)$.

Finally, assume that $(R,S)$ was added in \textbf{Case 2}. For every closed child $v_i$ of $v$ such that $\mathcal{R}(v_i)$ contains the record $(\{v_iv\}, \emptyset)$, we pick a partial solution $D_i$ witnessing this record. 
For the rest of closed children $v_i \in C$, if $s<V_v(\{v_iv\})$, let $D_i$ be the partial solution at $v_i$ witnessing the record $(\{vv_i\},\{v_i\})$, and otherwise let $D_i$ be the witness of $(\{vv_i\},\emptyset)$.
\end{proof}

We obtain $D$ by gluing together all $D_i$, $i\in [c]_0$, i.e., by taking the union of their vertices and arcs. Note that if some arc belongs to more than one $D_i$, $i \in [t]_0$, then it is also present in $R'$. Since $R'$ is anti-symmetric and the union of boundaries of open children along with $V_0$ covers the boundary at $v$ by Point 3. of Observation~\ref{obs:basiclfesproperties}, $D$ defines some orientation of all the edges incident to $\bar V_v$. In particular, $R'$ contains an arc for each edge of $G$ with precisely one endpoint in $V_v$, so $R$ satisfies Point. 2 of Definition \ref{def:record} by construction.


Towards Point. 4, consider any vertex $w \in S=S' \cap \delta (v)$. If $w \in S_0$, then all the arcs of $D$ incident to $w$ belong to $R'$ and by our choice of $S_0$ in CRC($v$)there is precisely one such arc. Assume that $w \in V_i$ for some open child $v_i$ of $v$. Then $w$ also belongs to $S_i$ since CRC($v$)ensures that $S_i \cap V_i =S' \cap V_i$. Hence, the in-degree of $w$ in $D_i$ is equal to one. As $w \in V_i$, $D_i$ contains all the arcs of $D$ incident to $w$, so the in-degree of $w$ in $D$ is equal to one as well. Finally, if $w \in S_i$ for some $i \in [t]$ but $w$ does not belong to any $V_j$, $j \in [t]$, then all the arcs of $D$ incident to $w$ belong to $R'$, and $R'$ contains at most one incoming arc to each vertex in $S'$. On the other hand, $w$ belongs to $S_i$, so there should be an incoming arc to $w$ in $D_i$ by Point. 4 of Definition \ref{def:record}. Hence, the in-degree of $w$ in $D$ is equal to one. 

To check Point 1 of Definition \ref{def:record}, we will ensure that the allocation defined by $D$ is EFX when restricted to $V_v$. Since $D_i$, $i\in [c]$, are partial solutions, it is EFX between any two vertices that belong to the same $V_i$, $i\in [c]$. If $v$ envies some $v_i$ in $C$, recall that for such $v_i$ we chose $D_i$ to be a witness of $(\{vv_i\}, \{v_i\})$, which means that the in-degree of $v_i$ in $D_i$ is equal to one. Moreover, if the record was added in the branch corresponding to \textbf{Case 0} or \textbf{Case 1}, there is no envy from any closed child towards $v$, since none of their incident edges are oriented towards $v$. If there is envy from $v_i$ towards $v$ in \textbf{Case 1'}, then $v$ has no incoming arcs in $R'$, so the envy can be eliminated by removing the unique incoming arc $v_iv$ of $v$. Finally, if $(R,S)$ was added in the branch corresponding to \textbf{Case 2}, then for every closed child $v_i$ such that $v_iv$ is oriented towards $v$, recall that we chose $D_i$ as a witness of $(\{v_iv\}, \emptyset)$, so $v_i$ does not envy $v$.
Hence, the allocation is EFX between $v$ and any $v_i \in C$. 

Consider a pair of adjacent vertices $u$ and $w$ in $V_v$ such that $u \in V_i$ and $w\in V_j$ for some open children $v_i$ and $v_j$ of $v$. Assume that there is envy from $u$ towards $w$ in $D$, then this envy is present also in $D_i$. By the definition of the partial solutions, we conclude that $w$ belongs to $S_i \subseteq S'$. Since $S_j \cap V_j =S' \cap V_j$, $w$ belongs to $S_j$ as well and therefore has in-degree one in $D_j$. As $w\in V_j$, there are no incoming arcs to $w$ outside of $D_j$ and hence the in-degree of $w$ in $D$ is equal to one. Therefore, the allocation is EFX between $u$ and $w$. 

It remains to consider $v$ and any neighbor $w$ of $v$ in $V_v\setminus C$. Observe that $w\in V_i \cap \delta (v_i)$ for some $i\in [t]$. If $v$ envies $w$, CRC($v$) ensures that $w$ belongs to $S'$, otherwise, the branch would be discarded after comparing the value of $v$ with $\mathcal V_v(\{vw\})$. As $S' \cap V_i = S_i \cap V_i$, this means that $w \in S_i$, so the in-degree of $w$ in $D_i$ (and hence in $D$, as $w\in V_i$) is equal to one. Conversely, if $w$ envies $v$, then since $D_i$ is a partial solution at $v_i$, we conclude that $v\in S_i\subseteq S'$, in particular, $v$ has no incoming arcs other than $wv$ in $R'$ and $(R,S)$ was added in the branch of CRC($v$)corresponding to \textbf{Case 1}. By our choice of witnesses for the closed children, $v$ does not receive incoming arcs from any of them, so its in-degree in $D$ is equal to one. Hence, the orientation defined by $D$ is EFX if restricted to $V_v$. 

Finally, to check Point 3 of Definition \ref{def:record}, consider any pair of adjacent vertices $u$ and $w$ of $D$ such that $u\in V_v$, $w \in \delta(v)$ and $u$ envies $w$. If $u$ belongs to $V_i$ for some $i\in [t]$, then the envy is present also in $D_i$ and hence $w \in S_i \subseteq S'$. Otherwise, if $u\in V_v$ envies $w \in \delta(v)$ and $u$ does not belong to any of $V_i$, $i\in [t]$, we conclude that $u=v$. Then CRC($v$) ensures that $w$ belongs to $S'$, otherwise, the branch would be discarded after comparing the value of $v$ with $\mathcal V_v(\{vw\})$. As $w \in  \delta(v)$, it belongs to $S$. 
\fi
This concludes the proof of Lemma~\ref{lem:combRecordsnew}.
\end{proof}

To complete the proof of Theorem \ref{thm: EFX_FPT}, we first compute the
records $\mathcal R(v)$ for each leaf of $T$ via exhaustive branching. Then we apply Lemma \ref{lem:combRecordsnew} to propagate our
record sets towards the root of $T$. 


\section{Discussion and Open Problems}
\label{sec:conclusion}
The focus of this paper was the existence and complexity of EF1 and EFX orientations of goods, i.e., allocations that satisfy the corresponding fairness criterion and in addition, every agent gets items from their own predetermined set. 
In contrast to EFX orientations, which do not always exist and if they do it is almost always hard to find one, we have shown that EF1 orientations do exist and can be computed. 
Hence, EF1 orientations, in addition to fairness constraints, preserve some economic efficiency as well; however in~\cite{christodoulou2023fair} it was shown that EFX does not have the latter property, as in specific cases it {\em has} to give goods to agents that do not value them at all.

We conclude by highlighting some intriguing open problems that deserve further investigation.


    From a purely theoretical point of view, the complexity of finding an EF1 orientation is open. Can we compute such an orientation in polynomial time? We currently do not know the answer even for three agents- for two agents the problem is easy. 
    Alternatively, is the problem hard for some complexity class? Our proof indicates that the problem belongs to \PLS, however, showing hardness is an intriguing question.
    
    What about EF1 allocations that satisfy other types of constraints? Our results show existence under ``approval'' constraints for the agents. A different option would be to consider {\em cardinality} constraints for the agents, i.e., every agent should get a specific number of items. A version of this model was studied by \citeauthor{caragiannis2024repeatedly}~\shortcite{caragiannis2024repeatedly} and the existence of EF1 allocations is an open problem.
    
    The definition of EFX that both \citeauthor{christodoulou2023fair}~\shortcite{christodoulou2023fair} and this paper has adopted, assumes that the envy should be eliminated by removing any item from an envied bundle. However, someone can study the relaxed notion of EFX0; recall in EFX0 the envy should be eliminated by removing any item that changes the value of the envied bundle. 
    Our preliminary investigation shows that indeed this is a promising direction. We have identified some classes of (graph-based) monotone valuations where a natural generalization of EFX0 always exists and can be computed efficiently. This raises the natural question about for which classes of valuation functions EFX0 allocations are guaranteed to exist.

\bibliographystyle{plainnat}
\bibliography{references}

\begin{thebibliography}{22}
\providecommand{\natexlab}[1]{#1}
\providecommand{\url}[1]{\texttt{#1}}
\expandafter\ifx\csname urlstyle\endcsname\relax
  \providecommand{\doi}[1]{doi: #1}\else
  \providecommand{\doi}{doi: \begingroup \urlstyle{rm}\Url}\fi

\bibitem[Akrami et~al.(2022{\natexlab{a}})Akrami, Alon, Chaudhury, Garg, Mehlhorn, and Mehta]{akrami2022efx}
Hannaneh Akrami, Noga Alon, Bhaskar~Ray Chaudhury, Jugal Garg, Kurt Mehlhorn, and Ruta Mehta.
\newblock {EFX} allocations: Simplifications and improvements.
\newblock \emph{arXiv preprint arXiv:2205.07638}, 2022{\natexlab{a}}.

\bibitem[Akrami et~al.(2022{\natexlab{b}})Akrami, Rezvan, and Seddighin]{ijcai2022p3}
Hannaneh Akrami, Rojin Rezvan, and Masoud Seddighin.
\newblock An {EF2X} allocation protocol for restricted additive valuations.
\newblock In \emph{Proceedings of the Thirty-First International Joint Conference on Artificial Intelligence, {IJCAI'22}}, pages 17--23. International Joint Conferences on Artificial Intelligence Organization, 2022{\natexlab{b}}.
\newblock \doi{10.24963/ijcai.2022/3}.

\bibitem[Amanatidis et~al.(2023)Amanatidis, Aziz, Birmpas, Filos{-}Ratsikas, Li, Moulin, Voudouris, and Wu]{AmanatidisABFLMVW2023}
Georgios Amanatidis, Haris Aziz, Georgios Birmpas, Aris Filos{-}Ratsikas, Bo~Li, Herv{\'{e}} Moulin, Alexandros~A. Voudouris, and Xiaowei Wu.
\newblock Fair division of indivisible goods: Recent progress and open questions.
\newblock \emph{Artificial Intelligence}, 322:\penalty0 103965, 2023.
\newblock \doi{10.1016/j.artint.2023.103965}.

\bibitem[Bouveret et~al.(2017)Bouveret, Cechl{\'a}rov{\'a}, Elkind, Igarashi, and Peters]{bouveret2017fair}
S~Bouveret, K~Cechl{\'a}rov{\'a}, E~Elkind, A~Igarashi, and D~Peters.
\newblock Fair division of a graph.
\newblock In \emph{Proceedings of the 26th International Joint Conference on Artificial Intelligence, {IJCAI}~'17}, pages 135--141, 2017.

\bibitem[Bouveret and Lang(2008)]{BouveretL08}
Sylvain Bouveret and J{\'{e}}r{\^{o}}me Lang.
\newblock Efficiency and envy-freeness in fair division of indivisible goods: Logical representation and complexity.
\newblock \emph{Journal of Artificial Intelligence Research}, 32:\penalty0 525--564, 2008.

\bibitem[Budish(2011)]{Budish11}
Eric Budish.
\newblock The combinatorial assignment problem: Approximate competitive equilibrium from equal incomes.
\newblock \emph{Journal of Political Economy}, 119\penalty0 (6):\penalty0 1061--1103, 2011.

\bibitem[Caragiannis and Narang(2024)]{caragiannis2024repeatedly}
Ioannis Caragiannis and Shivika Narang.
\newblock Repeatedly matching items to agents fairly and efficiently.
\newblock \emph{Theoretical Computer Science}, 981:\penalty0 114246, 2024.

\bibitem[Caragiannis et~al.(2019{\natexlab{a}})Caragiannis, Gravin, and Huang]{caragiannis2019envy}
Ioannis Caragiannis, Nick Gravin, and Xin Huang.
\newblock Envy-freeness up to any item with high nash welfare: The virtue of donating items.
\newblock In \emph{Proceedings of the 2019 ACM Conference on Economics and Computation}, pages 527--545, 2019{\natexlab{a}}.

\bibitem[Caragiannis et~al.(2019{\natexlab{b}})Caragiannis, Kurokawa, Moulin, Procaccia, Shah, and Wang]{CaragiannisKMPSW19}
Ioannis Caragiannis, David Kurokawa, Herv{\'e} Moulin, Ariel~D Procaccia, Nisarg Shah, and Junxing Wang.
\newblock The unreasonable fairness of maximum {N}ash welfare.
\newblock \emph{{ACM} Transactions on Economics and Computation}, 7\penalty0 (3):\penalty0 1--32, 2019{\natexlab{b}}.

\bibitem[Chaudhury et~al.(2021)Chaudhury, Kavitha, Mehlhorn, and Sgouritsa]{chaudhury2021little}
Bhaskar~Ray Chaudhury, Telikepalli Kavitha, Kurt Mehlhorn, and Alkmini Sgouritsa.
\newblock A little charity guarantees almost envy-freeness.
\newblock \emph{SIAM Journal on Computing}, 50\penalty0 (4):\penalty0 1336--1358, 2021.

\bibitem[Chaudhury et~al.(2024)Chaudhury, Garg, and Mehlhorn]{chaudhury2024efx}
Bhaskar~Ray Chaudhury, Jugal Garg, and Kurt Mehlhorn.
\newblock {EFX} exists for three agents.
\newblock \emph{Journal of the ACM}, 71\penalty0 (1):\penalty0 1--27, 2024.

\bibitem[Christodoulou et~al.(2023)Christodoulou, Fiat, Koutsoupias, and Sgouritsa]{christodoulou2023fair}
George Christodoulou, Amos Fiat, Elias Koutsoupias, and Alkmini Sgouritsa.
\newblock Fair allocation in graphs.
\newblock In \emph{Proceedings of the 24th ACM Conference on Economics and Computation}, pages 473--488, 2023.

\bibitem[Cygan et~al.(2015)Cygan, Fomin, Kowalik, Lokshtanov, Marx, Pilipczuk, Pilipczuk, and Saurabh]{CyganFKLMPPS2015}
Marek Cygan, Fedor~V. Fomin, {\L{}}ukasz Kowalik, Daniel Lokshtanov, D{\'{a}}niel Marx, Marcin Pilipczuk, Micha{\l{}} Pilipczuk, and Saket Saurabh.
\newblock \emph{Parameterized Algorithms}.
\newblock Springer, Cham, 2015.
\newblock ISBN 978-3-319-21274-6.
\newblock \doi{10.1007/978-3-319-21275-3}.

\bibitem[Deligkas et~al.(2021)Deligkas, Eiben, Ganian, Hamm, and Ordyniak]{DeligkasEGHO2021}
Argyrios Deligkas, Eduard Eiben, Robert Ganian, Thekla Hamm, and Sebastian Ordyniak.
\newblock The parameterized complexity of connected fair division.
\newblock In \emph{Proceedings of the 30th International Joint Conference on Artificial Intelligence, {IJCAI}~'21}, pages 139--145, 2021.
\newblock \doi{10.24963/ijcai.2021/20}.

\bibitem[Ganian and Korchemna(2021)]{GanianK2021}
Robert Ganian and Viktoriia Korchemna.
\newblock The complexity of bayesian network learning: Revisiting the superstructure.
\newblock In \emph{Proceedings of the 35th Conference on Neural Information Processing Systems, {NeurIPS~'21}}, pages 430--442, 2021.

\bibitem[Ganian and Korchemna(2024)]{ganian2024slim}
Robert Ganian and Viktoriia Korchemna.
\newblock Slim tree-cut width.
\newblock \emph{Algorithmica}, pages 1--25, 2024.

\bibitem[Goldberg et~al.(2023)Goldberg, H{\o}gh, and Hollender]{goldberg2023frontier}
Paul~W Goldberg, Kasper H{\o}gh, and Alexandros Hollender.
\newblock The frontier of intractability for {EFX} with two agents.
\newblock In \emph{International Symposium on Algorithmic Game Theory}, pages 290--307. Springer, 2023.

\bibitem[Kyropoulou et~al.(2020)Kyropoulou, Suksompong, and Voudouris]{kyropoulou2020almost}
Maria Kyropoulou, Warut Suksompong, and Alexandros~A Voudouris.
\newblock Almost envy-freeness in group resource allocation.
\newblock \emph{Theoretical Computer Science}, 841:\penalty0 110--123, 2020.

\bibitem[Lipton et~al.(2004)Lipton, Markakis, Mossel, and Saberi]{LiptonMMS04}
Richard~J Lipton, Evangelos Markakis, Elchanan Mossel, and Amin Saberi.
\newblock On approximately fair allocations of indivisible goods.
\newblock In Jack~S. Breese, Joan Feigenbaum, and Margo~I. Seltzer, editors, \emph{Proceedings of the 5th ACM Conference on Electronic Commerce, {EC}~'04}, pages 125--131. {ACM}, 2004.
\newblock \doi{10.1145/988772.988792}.

\bibitem[Payan et~al.(2023)Payan, Sengupta, and Viswanathan]{payan2023relaxations}
Justin Payan, Rik Sengupta, and Vignesh Viswanathan.
\newblock Relaxations of envy-freeness over graphs.
\newblock \emph{AAMAS 2023}, pages 2652--2654, 2023.

\bibitem[Plaut and Roughgarden(2020)]{PlautR20}
Benjamin Plaut and Tim Roughgarden.
\newblock Almost envy-freeness with general valuations.
\newblock \emph{{SIAM} J. Discret. Math.}, 34\penalty0 (2):\penalty0 1039--1068, 2020.
\newblock \doi{10.1137/19M124397X}.

\bibitem[Zhou et~al.(2024)Zhou, Wei, Li, and Li]{zhou2024EFX}
Yu~Zhou, Tianze Wei, Minming Li, and Bo~Li.
\newblock A complete landscape of {EFX} allocations of mixed manna on graphs.
\newblock \emph{arXiv preprint arXiv:2409.03594}, 2024.

\end{thebibliography}

\iflong

\newpage
\appendix

\section{EFXr: EFX for relevant items}

Using the notion of relevant items we can define the following relaxation of EFX, where we remove only those items and we term it EFXr.

\begin{definition}[EFXr]
An allocation $\pi$ is envy-free up to any {\em relevant} item (EFXr), if for every pair of agents $i,j \in N$ and every item $a \in \pi_j$ that is relevant for $i$ it holds that $\mathcal V_i(\pi_i) \geq \mathcal V_i(\pi_j\setminus a)$.
\end{definition}

\subsection{Decomposable instances}
\begin{definition}
 We say that the instance $\mathcal I$ of fair allocation with agent lists $N_a$, $a\in A$, is \emph {decomposable} if for every pair $a, b$ of items it holds that either $N_a=N_b$ or $|N_a \cap N_b|=1$.   
\end{definition}
In other words, the instance is decomposable if items from different groups share at most one agent. In the following lemma, we show that to obtain EF1 (EFXr, EF) orientation of a decomposable instance, it is sufficient to compute EF1 (EFXr, EF) allocation for each item-group independently. 
 
 Let us denote by $\mathcal I^g$ the restriction of $\mathcal I$ where the agents are the same, but the only available items are from the group $g$. Valuations are naturally restricted to this set of items.

\begin{lemma}
\label{lem: decomposable}
 If $\mathcal{I}$ is decomposable instance of fair allocation with monotone valuations and EF1 (EFX, EF) allocation $\pi^g$ exists for every $\mathcal I^g$, then their union $\pi$ with $\pi_i = \bigcup_g \pi_i^g$, $i\in N$, is EF1 (EFXr, EF) allocation for $\mathcal I$.
\end{lemma}
\begin{proof}
  Consider any pair of agents $i$ and $j$ and assume that there is envy from $i$ towards $j$, i.e $\mathcal V_i(\pi_j)>\mathcal V_i(\pi_i)$. Due to decomposability, there is at most one item group $g$ of items relevant for both $i$ and $j$. In particular, the items of $j$ outside of $g$ are not relevant for $i$ and hence $\mathcal V_i(\pi_j)=\mathcal V_i(\pi_j^g)$. Moreover, since the valuations are monotone, it holds that $\mathcal V_i(\pi_i)\ge \mathcal V_i(\pi_i^g)$. Hence $\mathcal V_i(\pi_j^g)=\mathcal V_i(\pi_j)>\mathcal V_i(\pi_i) \ge \mathcal V_i(\pi_i^g)$, so the envy from $i$ towards $j$ is present in $\pi^g$. 
  
  Consider any item $a$ such that its removal eliminates this envy in $\pi^g$, then $a$ belongs to the group $g$ and $\mathcal V_i(\pi_j^g \setminus \{a\})\leq \mathcal V_i(\pi_i^g)$. Since items of $j$ that do not belong to $g$ are not relevant for $i$, we have $\mathcal V_i(\pi_j \setminus \{a\}) = \mathcal V_i(\pi_j^g \setminus \{a\})\leq \mathcal V_i(\pi_i^g) \leq \mathcal V_i(\pi_i)$, so removal of $a$ eliminates envy from $i$ towards $j$ in $\pi$. In particular, this shows that if $\pi^g$ is EF1 between $i$ and $j$, so is $\pi$.
  
  Finally, assume that $i$ envies $j$ in $\pi$ and $\pi^g$ is EFX. Let $a$ be any item in $\pi_j$ relevant for $i$, then $a$ belongs to $\pi^g_j$. As $\pi^g$ is EFX, removal of $a$ eliminates envy from $i$ towards $j$ in $\pi^g$. Therefore, as we argued in the last paragraph, the removal of $a$ eliminates envy from $i$ towards $j$ in $\pi$. Hence, $\pi$ is \EFXr between $i$ and $j$. Applying these arguments to all pairs of agents proves the lemma.
\end{proof}

Multigraphs provide a simple example of decomposable instances: item-groups can be identified with edges, and every pair of edges share at most one vertex. Plaut and Roughgarden \cite{PlautR20} showed that EFX valuations exist for two agents with arbitrary valuations. This allows us to obtain an EFX allocation for each edge of the multigraph independently. By Lemma \ref{lem: decomposable}, gluing them together results in \EFXr orientation for the whole multigraph.

\begin{corollary}
 \EFXr orientation always exists for the multigraphs with monotone valuations. It can be computed in poly-time if an EFX allocation for each multi-edge is provided. 
\end{corollary}

\subsection{\EFXr orientation of faces of planar graphs}
In this section, we show the existence of \EFXr (and hence EF1) orientations for the special class of instances $\mathcal{I}$ of Fair Division when the agents can be identified with vertices of some embedded planar graph $G=G_{\mathcal I}$, while each internal face of $G$ represents a unique item that is relevant for those and only those agents which belong to the face.

\begin{theorem}
\label{thm:planargraphs}
Given an embedded planar graph $G$ with triangulated inner faces and monotone valuations, an \EFXr orientation of its faces can be obtained in polynomial time. 
\end{theorem}
\begin{proof}
    
Without loss of generality, we assume that $G$ is 2-connected, otherwise, we can combine allocations of its 2-connected components. We will allocate the inner faces of $G$ to its vertices so that the following conditions hold:
\begin{enumerate}
\item At most 2 faces are allocated to each vertex,
\item For every pair of vertices sharing inner faces, at most one common face is allocated to each vertex of the pair.
\end{enumerate}

We call the allocations satisfying these conditions \emph{proper}.
Assume that $\pi$ is a proper allocation.
Then, if there is envy from some vertex $v$ towards another vertex $w$, it can be eliminated by removing the unique shared inner face of $v$ and $w$ from $\pi_w$. Moreover, $\pi$ is \EFXr, since the rest of the faces in $\pi_w$ are not relevant to $v$. 

We will show how to explicitly construct proper allocations by induction on the number of vertices in $G$. 
If $G$ has only 3 vertices, we simply allocate the inner face to any of them. Assume that proper allocations can be constructed for graphs with $i$ vertices and $G$ has $i+1$ vertex. 
 
Let $v$ be the vertex of the boundary of $G$ such that $G \setminus \{v\}$ is 2-connected.
If $v$ has only two neighbors in $G$, there is a unique inner face $f$ containing $v$. We apply the inductive hypothesis to obtain the proper allocation $\pi'$ for $G\setminus {v}$ and then extend it by allocating $f$ to $v$. Otherwise,
let $v'$ be the neighbor of $v$ on the boundary of $G$. Then $v$ and $v'$ share precisely one inner face, and its third vertex $v''$ does not belong to the boundary of $G$. We denote the face by $f_0=vv'v''$.

We construct the graph $G'$ from $G$ as follows. Delete the vertex $v$ and add the edges from $v'$ to every neighbor of $v$ in $G$, unless they were connected before. We say that the face $f'$ of $G'$ is the \emph{image} of the face $f$ of $G$ if $f$ contains $v$, $f'$ contains $v'$, and the rest of their vertices coincide. Note that $f_0$ is the only one face of $G$ that is neither preserved in $G'$ nor has an image in $G'$.

By construction, $G'$ is a planar 2-connected graph on $i$ vertices with triangulated inner faces. By inductive hypothesis, there exists a proper allocation $\pi'$ of inner faces of $G'$. We use $\pi'$ to construct a proper allocation $\pi$ of inner faces of $G$. First, we let $\pi$ coincide with $\pi'$ on the faces that appear both in $G$ and $G'$. Moreover, for each inner face $f$ that has an image $f'$ in $G'$, $\pi$ allocates $f$ to the same vertex to which $f'$ is allocated by $\pi'$ unless it is $v'$. In case $\pi'(f')=v'$, we let $\pi(f)=v$. Finally, if $v$ receives less than two faces and does not receive the face $f_1\neq f_0$ containing the edge $vv''$, we set $\pi(f_0)=v$, and otherwise $\pi(f_0)=v'$. It is clear that $\pi$ respects the agent lists.

 Let us ensure that $\pi$ is proper. Note that each vertex other than $v$ and $v'$ receives at most two faces in $G'$ and hence in $G$. If $\pi(f_0)=v$, then by construction $\pi$ gives at most two faces to $v$, while $v'$ receives the same or a smaller number of faces compared to $\pi'$. If $\pi(f_0)=v'$, observe that $\pi$ gives to $v$ at most as many faces of $G$ as $\pi'$ gives to $v'$. At the same time, the number of faces allocated to $v'$ by $\pi$ is at most the same as for $\pi'$, although $v'$ receives a new face $f_0$, $\pi$ allocates to $v$ at least one face $f$ such that its image $f'$ was allocated to $v'$ by $\pi'$.

We show that the second condition holds as well. Let $u\ne v$ and $w\ne v$ be two different vertices of $G$ such that at least one of them does not belong to $f_0$. If $u$ and $w$ share faces in $G$, they share the same faces or their images in $G'$. The second condition is satisfied in $G'$ for $u$ and $w$, and hence it holds in $G$ as by constructing $\pi$ from $\pi'$ we can only decrease the numbers of allocated shared faces for $u$ and $w$. 

Now consider a pair consisting of $v$ and $u$ such that $u$ shares with $v$ two inner faces but does not belong to $f_0$. Then $v'$ and $u$ share images of these faces in $G'$. As $\pi'$ is proper, it allocates at most one of them to $v'$ and at most one to $u$. By construction, $\pi$ allocates at most one of the original faces to $v$ and $u$ correspondingly.  

It remains to ensure that the second condition holds for pairs of vertices that both belong to $f_0$ and share at least two inner faces, i.e. for $v,v''$ and $v',v''$. For $v$ and $v''$, recall that by construction only one of their shared faces $f_0,f_1$ can be allocated to $v$. Same holds for $v''$, as $\pi(f_0)\neq v''$. Moreover, for $v'$ and $v''$, if $\pi(f_0)=v$, the second condition is satisfied. Otherwise, if $\pi(f_0)=v'$, we conclude that either $\pi(f_1)=v$ or $\pi$ gives to $v$ two faces other than $f_1$. If $\pi(f_1)=v$, by construction $\pi'(f_1')=v'$, so $\pi'$ can not allocate to $v'$ its second shared with $v''$ face. Finally, if $\pi$ gives to $v$ two faces other than $f_1$, then $\pi'$ gives their images to $v'$. By the first condition, no other faces can be allocated by $\pi'$ to $v'$, in particular no faces shared with $v''$. Hence, $f_0$ is the only shared face of $v'$ and $v''$ allocated by $\pi$ to $v'$.

\end{proof}

\fi

\end{document}